\documentclass[12pt]{article} 
\usepackage[utf8]{inputenc}
\usepackage[a4paper]{geometry}
\usepackage[T1]{fontenc}
\usepackage{setspace}
\usepackage{amsthm,amsmath}
\usepackage{amssymb}
\usepackage{natbib}
\usepackage[colorlinks=true,linkcolor=black, citecolor=black, urlcolor=black]{hyperref}
\usepackage{bbm}
\usepackage{multicol}
\usepackage{multirow}
\usepackage{mathtools}
\usepackage{url}          
\usepackage{float}
\usepackage{bm}
\usepackage{xcolor}
\usepackage{graphicx} 
\usepackage{caption} 
\usepackage{subcaption} 
\usepackage{algpseudocode}
\usepackage{algorithm}
\usepackage[noEnd=false]{algpseudocodex}

\usepackage{booktabs}
\usepackage{soul}
\usepackage{bbm}

\newcommand{\rev}[1]{{\color{black}  #1}}

\newcommand{\VaR}{\mathrm{VaR}}
\newcommand{\ES}{\mathrm{ES}}
\newcommand{\MES}{\mathrm{MMES}}
\newcommand{\DCTE}{\mathrm{DCTE}}

\newtheorem{theorem}{Theorem}[section]

\newtheorem{proposition}[theorem]{Proposition}

\theoremstyle{definition}

\title{  {Non-Parametric Simulation of Multivariate Extreme Events via Spectral Bootstrap}}
\author{Nisrine Madhar$^{1}$, Juliette Legrand$^{2}$ and Maud Thomas$^{3,4}$}
\date{}

\begin{document}

\maketitle

{\small \noindent
$^1$ Universit\'e Paris Cit\'e, CNRS, Laboratoire de Probabilit\'es, Statistique et Mod\'elisation, LPSM, F-75013 Paris, France,\\
$^2$ Univ Brest, CNRS UMR 6205, Laboratoire de Mathématiques de Bretagne Atlantique, France \\
$^3$ Sorbonne Universit\'e, CNRS, Laboratoire de Probabilit\'es, Statistique et Mod\'elisation, LPSM, 4 place Jussieu, F-75005 Paris, France,\\
$^4$ Universit\'e Claude Bernard Lyon 1, Laboratoire de Sciences Actuarielle et Financi\`ere (UR SAF), Lyon, 69007, France

E-mails: madhar@lpsm.paris, juliette.legrand@univ-brest.fr, maud.thomas@univ-lyon1.fr}

\begin{abstract}

  {Inference in extreme value theory relies on a limited number of extreme observations, making estimation challenging. To address this limitation, we propose a non-parametric \rev{simulation scheme}, the \textit{multivariate extreme events spectral bootstrap} \rev{simulation}  {procedure}, relying on the spectral representation of multivariate generalized Pareto-distributed random vectors. Unlike standard bootstrap methods, our approach preserves the joint tail behaviour of the data and generates additional synthetic extreme data, thereby improving the reliability of inference. We demonstrate the effectiveness of our procedure \rev{on} the estimation of tail risk metrics, under both simulated and real data. The results highlight the potential of this method for enhancing risk assessment in high-dimensional extreme scenarios.}
\end{abstract}


\section{Introduction}\label{sec:Intro}

 {
Extreme value theory (EVT) provides the statistical framework for the study of extreme events. It allows statistical inference beyond the observed sample, such as estimating the probability of an event more extreme than those already observed, or evaluating an extreme quantile, that is, a quantile whose order tends to one as the sample size tends to infinity. These questions arise in many fields, such as finance (e.g. financial crises and Value-at-Risk), climatology (e.g. droughts and severe wind storms) and insurance (e.g. cyber risk and bodily injury). For example, one might wish to determine the height of a dyke so that the probability of a flood occurring is $10^{-4}$. Similarly, one might want to estimate the probability of observing a wave higher than the highest wave that has ever been recorded. EVT focuses on the tail of the sample distribution, that is to say, either on the largest order statistics or on the observations exceeding a certain threshold are considered. Therefore, inference is based only on a sub-sample, which can be rather small. On the other hand, the theory is essentially asymptotic. In particular, estimators are proven to have good asymptotic properties when the size of this subset of data tends toward infinity \citep[for more details, see][]{de2006extreme}. 

To overcome this issue, a natural idea is to use bootstrap procedures to increase the size of the sub-sample used for inference. \rev{However, bootstrap cannot be used as it is for EVT. The main reason is that empirical distribution does not satisfy the EVT conditions even though the underlying sample distribution does. Nonetheless, \citet{de2024bootstrapping} were able to derive a bootstrap version of the fundamental expansions for the tail quantile process in the univariate framework of EVT, allowing the bootstrap to mimics, for any estimator in EVT whose asymptotic properties can be established via the tail quantile process, the original behaviour of the estimator.}  

\rev{The main goal of this paper is to propose} a non-parametric approach to expand the number of observations above the extreme level of interest based on joint simulation of multivariate extremes, by extending the approach developed by \citet{legrand2023joint} in the bivariate case to higher dimensions. This  \rev{simulation} procedure makes use of a spectral representation of multivariate extreme vectors, as proven by \citet{ferreira2014generalized, rootzen2018multivariate2}. \rev{This spectral representation writes standard multivariate generalized Pareto vectors as the sum of a standard exponential random variable and a $d$-dimensional vector. The key idea is to simulate standard exponential variables and to use non-parametric resampling, or bootstrap, on the $d$-dimensional vector, independently. This allows to simulate new observations.} We refer to this procedure as the {\it multivariate extreme spectral bootstrap (SB) \rev{simulation} procedure}, or the {\it SB \rev{simulation} procedure} for short. 
\rev{This procedure inherits the property of consistency of the bootstrap \citep[see][]{bickel1981}. }
It enables us to increase the number of extreme observations and thus make more reliable estimations of quantities depending on the tail of a multivariate extreme distribution.

As an illustration, we apply our procedure to the estimation of tail risk metrics (TRM).  There has been a surge in the development of TRM estimation methodologies in modern risk management. The simplest and most common measure often used to evaluate tail risks is the Value-at-Risk (VaR) which is the quantile, usually of a (very) large order, of a given random variable. Common estimators are obtained as empirical versions of the TRM. However, for high-risk scenarios located in the tail, the number of available tail observations is extremely limited. We therefore apply our bootstrap-type procedure to increase the amount of data available in the tail.  This procedure is applied to both simulated and real data. 

The outline of this paper is as follows: after a brief introduction of EVT required for the \rev{SB simulation procedure} in Section~\ref{sec:MGP}, Section~\ref{sec:joint} details the \rev{SB simulation procedure} for multivariate extremes. Section~\ref{sec:illustation} is then dedicated to illustrating the \rev{SB simulation procedure} on the TRM estimations. First,  the  TRMs  are introduced in Section~\ref{sec:TRM}. The performance of the procedure is then  assessed first on simulated data in Section~\ref{subsec:NS} and then on real data in Section~\ref{sec:DataReal}.}

\paragraph{Notations.} Throughout the paper, we use the following notations for vectors. Symbols in bold denote $d$-dimensional vectors.  For example, a $d$-dimensional random vector is denoted by $\bm a = \left(a_1,\ldots,a_d\right)$, and $\bm a_{-j}$ denotes the vector $\bm a$ deprived from its $j$-th component for $j=1,\ldots,d$.  Operations and relations are meant component-wise, that is, for example, $\bm a\geq\bm b = \left(a_1\geq b_1,\ldots,a_d \geq b_d\right)$.

\section{Elements of  Extreme Value Theory}\label{sec:MGP}




 {This section presents the elements of EVT required for the SB simulation procedure.

EVT is a branch of statistics developed and widely used to model extreme events, such as major floods, heatwaves, or severe financial losses \citep{katz2002statistics, embrechts2013modelling}. EVT was originally developed between the 1920s and 1940s, with key contributions from \citet{frechet1927loi, fisher1928limiting, von1936distribution, gnedenko1943distribution}. In the 1970s, it was broadened with the Generalized Pareto (GP) framework, making it practical for real-world applications in risk assessment and natural hazard modeling. The multivariate EVT was systematically developed from the 1970s onward \citep{balkema1977max, de1977limit}, extending the univariate limit distributions to joint distributions.

Given that the \rev{SB simulation procedure} relies on multivariate EVT, we provide the essential components directly in the multivariate framework.
}

  {Let $\bm{X}$ be a random vector in $\mathbb{R}^d$.  {Introducing the appropriate limiting distribution for multivariate threshold exceedances requires considering the {\it maximum domain of attraction} of a multivariate extreme value distribution. } To this end, let $\bm{X}_1, \dots, \bm{X}_n$ be independent and identically distributed (i.i.d.) copies of $\bm{X}$. Suppose there exist two sequences $\bm{a}_n \in (0,\infty)^d$ and $\bm{b}_n \in \mathbb{R}^d$, and $G$ a multivariate c.d.f with non-degenerated margins such that 
\begin{equation}
\label{eq:max-domain}
    \lim_{n\to\infty} \mathbb{P}\left(\{\max_{1\leq i\leq n} \bm X_i -\bm b_n\} / \bm a_n\leq \bm x \right) = G(\bm x),
    \end{equation}
    for all $\mathbf x \in \mathbb R^d$ where $G$ is continuous. Then, $\bm X$ is said to be in the \textit{maximum domain of attraction} of the multivariate extreme value distribution $G$ \citep[e.g.][]{beirlant2004, de2006extreme, coles2001}. If \eqref{eq:max-domain} holds,  then there exists a distribution $H$ with non-degenerate margins such that \begin{equation}
        \label{eq:pot}
        \lim_{n\to\infty} \mathbb{P}\left(\{\bm X -\bm b_n\} /\bm a_n \leq \bm x \mid \bm X \nleq \bm b_n \right) = H(\bm x),
    \end{equation}
        for all $\mathbf x \in \mathbb R^d$ where $H$ is continuous. Here, $\bm X \nleq \bm b_n$  means that at least one component of $\bm X$ exceeds its corresponding threshold in $\bm b_n$.
    The limiting distribution $H$ in \eqref{eq:pot} belongs to the family of multivariate generalized Pareto (MGP) distributions. This motivates the modeling of  multivariate threshold exceedances using the MGP distributions. More specifically, if $\bm X$ is in the maximum domain of attraction of a multivariate extreme value distribution, then for a sufficiently high threshold $\bm u\in \mathbb R^d$, the conditional distribution of $\bm X - \bm u \mid \bm X \not \leq \bm u$ can be approximated by a MGP distribution  \citep{rootzen2006multivariate}.} 

 {

The result of \citet{rootzen2006multivariate} can be viewed as a multivariate extension of the fundamental theorem of \citet{balkema1974residual, pickands1975statistical} in the univariate case. Let $\zeta_1,\ldots,\zeta_n$ be $n$ real-valued i.i.d. random variables  with unknown distribution function $F$ and a threshold $u \in \mathbb R$. Then, if $F$ belongs to the maximum domain of attraction of a univariate generalized extreme value distribution, the conditional distribution $F_u$ of the excesses above $u$, that is $F_u(z) = \mathbb{P}\left(\zeta_1-u > z \mid \zeta_1>u \right)$
converges toward a univariate GP distribution, as $u$ tends to $\infty.$ }


  {For a MGP distribution $H$, its} marginals $H_j$, $j=1,\ldots, d$, are in general not univariate GP distributions. However, their restrictions to the positive subset are GP distributed \citep[see e.g.][]{rootzen2018multivariate2}, i.e. for all $j=1,\ldots, d$,
\begin{equation}\label{eq:MGPmargin}
H_j^+(x) = \mathbb{P} \left[Z_j > x \mid Z_j >0 \right] = \left(1 + \gamma_j x/\sigma_j\right)_+^{-1/\gamma_j}, \mbox{ for }  x\geq 0 \mbox{ such that } \sigma_j+\gamma_j x>0,
\end{equation}
where $a_+ = \max(a,0)$, $\sigma_j>0$ is the $j$-th marginal scale parameter and $\gamma_j \in \mathbb R$ is the $j$-th marginal shape parameter   {(also called the tail index of the $j$-th component)}. If $\gamma_j = 0$, we use the classical convention by taking the limit as $\gamma_j\to 0$ in Equation~\eqref{eq:MGPmargin}.   {This parameter}  {
reflects the heaviness of the tail, the larger $\gamma_j$, the heavier the tail of the $j$-th component. Three maximum sub-domains of attraction may be distinguished according to the sign $\gamma_j$. The case $\gamma_j>0$ corresponds to the heavy-tailed distributions, the case $\gamma_j=0$ to the light-tailed distributions and $\gamma_j<0$ to the right-finite tail distributions, \citep[for more details on the properties of the different maximum sub-domains of attraction, see e.g.]{coles2001, de2006extreme}}.  {The MGP distribution parameter vectors, denoted} $\bm \sigma = (\sigma_1,\ldots, \sigma_d)$ and $\bm \gamma = (\gamma_1,\ldots, \gamma_d)$,  {correspond to} the marginal scale and shape parameters. When $\bm \sigma = \bm 1$ and $\bm \gamma = \bm 0$, we say that $H$ is a standard MGP distribution.  {Note that any} standard MGP distributed vector $\bm Z$ can be transformed into a general MGP distributed vector $\bm Y$   {using the following one-to-one relation  $$\bm Y=\bm \sigma (e^{\bm\gamma \bm Z}-1)/\bm\gamma.$$} 
In light of this result, we will restrict our attention to the special case of standard MGP vectors.

 {\cite{rootzen2018multivariate2} have shown 
that a standard MGP vector $\bm Z$  can be decomposed as follows}
\begin{equation}\label{eq:SMGPRep}
    \bm Z = E + \bm T - \max_{1\leq k\leq d} T_k,
\end{equation}
where  $E$ is a unit exponential variable and $\bm T=(T_1,\dots,T_d)$ is any random vector in $\mathbb{R}^d$, independent of $E$,   {and such that $\max_{1\leq k\leq d}  T_k > -\infty$ almost surely and $\mathbb{P}(T_k > -\infty) > 0$ for all $k= 1, \dots,d$}.   {Such random vectors $\bm T$ are called spectral random vectors \citep[][Section 4]{rootzen2018multivariate2}, motivating the term  ``multivariate extreme spectral bootstrap''.}
From Equation~\eqref{eq:SMGPRep}, we derive the cornerstone of the \rev{SB simulation procedure}. Let us first define, for all $j=1,\dots,d$,


  {\begin{equation}
    \label{eq:wk_w1}
    \Delta_j:=Z_j-\max_{1\leq k\leq d} Z_k \, .
\end{equation}}
  {Then, Equation \eqref{eq:SMGPRep} can be rewritten as follows 
\begin{equation}\label{eq:MGPEq}
    \bm Z = E+\bm \Delta,
\end{equation}
where $\bm \Delta = \left(\Delta_1,\dots,\Delta_d\right)$ with the $\Delta_j$'s defined as in \eqref{eq:wk_w1}. With  this simple rewriting, we may proceed to the simulation of standard MGP vectors $\bm Z$, via the simulation of $\bm \Delta$  using bootstrapping techniques.}

 {Note that the \rev{SB simulation procedure} (Algorithm~\ref{alg:JointMGP}) relies on the spectral representation \eqref{eq:SMGPRep} for multivariate extreme vectors. An underlying assumption is thus that the components of $\bm X$ have non-trivial asymptotic tail dependence \citep[see e.g.][and Section 1 of Supplementary material]{coles1999}.} 


\section{Multivariate extreme spectral bootstrap procedure}\label{sec:joint}



  {This section introduces the multivariate extreme \rev{SB simulation procedure} developed in this study, extending the work of \citet{legrand2023joint}, who focused on the bivariate case. Leveraging the stochastic representation in \eqref{eq:MGPEq}, we propose a  {non-parametric} joint simulation algorithm,  {based on the spectral representation of MGP vectors, }  {to generate} samples from a MGP  {distribution}.
The performance of the procedure will be demonstrated through TRM estimation on synthetic data (Section~\ref{subsec:NS}) and real-world data (Section~\ref{sec:DataReal}).}

For a given sample $\bm Z_1,\dots,\bm Z_n$ of i.i.d. copies of a MGP vector $\bm Z$ in $\mathbb{R}^d$, the  {SB procedure} (see Algorithm~\ref{alg:JointMGP}) performs the stochastic generation of $m\in\mathbb{N}$ new i.i.d. copies of $\bm Z$, denoted $\bm Z_1^\ast,\dots,\bm Z_m^\ast$. 

Algorithm~\ref{alg:JointMGP} is as follows. First, we simulate unit exponential random variables. Then, independently on this first step, we simulate new realisations   {$\bm \Delta_1^\ast,\dots, \bm \Delta_m^\ast$ of the differences $\bm \Delta$} using a non-parametric bootstrap approach, i.e. by resampling among the observed differences  $\bm \Delta_{1},\dots,\bm\Delta_n$. Finally, we use the stochastic relation in Equation~\eqref{eq:MGPEq} to merge both simulation steps, in order to generate new realisations of MGP vectors.

\begin{algorithm}
\caption{Multivariate extreme spectral bootstrap \rev{simulation} algorithm}
\label{alg:JointMGP}
\begin{algorithmic}[1]
\State  \textbf{require}   {$n$ i.i.d. observations  $\bm Z_1,\dots,\bm Z_n$} of a standard  MGP distribution in $\mathbb{R}^d$
\State  \textbf{output}   {$m$ standard MGP simulated samples $\bm Z_1^\ast,\dots, \bm Z_m^\ast$ in $\mathbb{R}^d$}
\Procedure{}{} 
\State   {Compute $\Delta_{i,j} \leftarrow Z_{i,j}-\max_{1\leq k\leq d}Z_{i,k}$}, for $1\leq i\leq n$ and $1\leq j \leq d$
\State Generate  $E_1,\ldots, E_m \overset{\text{i.i.d.}}{\sim} \mathrm{Exp}(1)$, independently of $\bm \Delta_{1},\dots,\bm\Delta_n$ ~\\where $\bm\Delta_i=\left(\Delta_{i,1},\dots,\Delta_{i,d}\right)$ for $1\leq i\leq n$
\State Generate   {$m$ bootstrap samples}  $\bm \Delta_1^\ast,\dots, \bm \Delta_m^\ast$ from $\left(\bm \Delta_{1},\dots,\bm\Delta_n\right)$ 
\EndProcedure
\State  \textbf{return}   {$\bm Z_{\ell}^\ast \leftarrow E_\ell + \bm\Delta_\ell^\ast$ for  $1 \leq \ell \leq m$}
\end{algorithmic}
\end{algorithm}

 {We illustrate the performances of Algorithm~\ref{alg:JointMGP} with a numerical example involving a 3-dimensional MGP vector $\bm{Z}$. The vector $\bm{T}$ (as defined in Equation~\eqref{eq:SMGPRep}) is chosen as a centred multivariate Gaussian vector with correlation coefficients $\rho_{1,2} = 0.4$, $\rho_{1,3} = 0.8$, and $\rho_{2,3} = 0.1$. We simulate a sample of the 3-dimensional vector $\bm Z$ of size $2,000$ and a bootstrap sample $\bm Z^*$ of size $10000$. 

Figure~\ref{fig:Jointk3_scatter} shows the scatter plots of ($Z_1$, $Z_2$) in Figure a), ($Z_1$, $Z_3$) in Figure b) and ($Z_2$, $Z_3$) in Figure c) of the simulated sample in black crosses. On each plot, the bootstrap sample ($Z_1^*$, $Z_2^*$) in Figure a), ($Z_1^*$, $Z_3^*$) in Figure b) and ($Z_2^*$, $Z_3^*$) in Figure c) are overlaid (in red circles). These graphs show that the shape of the bootstrap sample seems to be similar to that of the simulated sample. Besides, it can be seen that some red circles (bootstrap sample) are above than the black crosses (simulated sample) for all components. This indicates that the bootstrap sample comprises a greater number of observations in the extreme regions, and that Algorithm ~\ref{alg:JointMGP} extends the simulated sample to include observations with more extreme values than those in the simulated sample. 

\begin{figure}
 \centering
   \begin{tabular}{ccc}
 
  \includegraphics[width=0.3\textwidth]{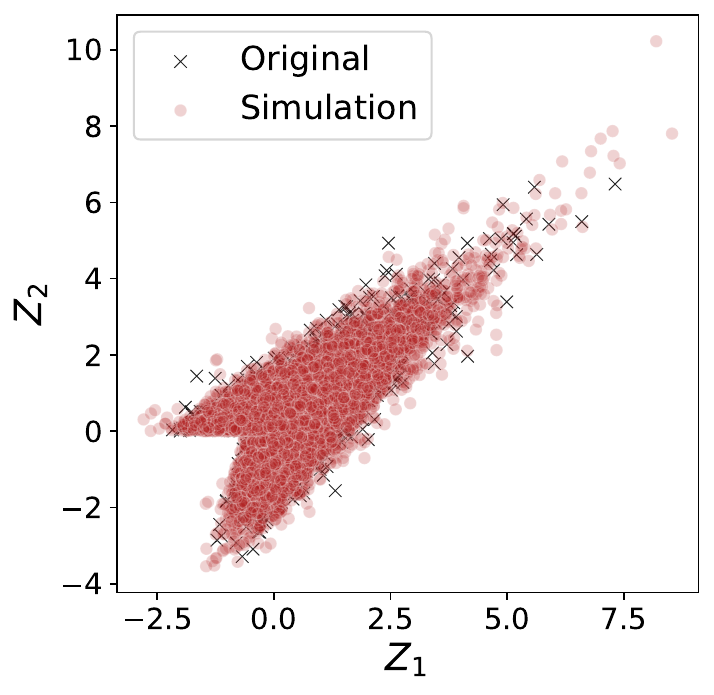}
& \includegraphics[width=0.3\textwidth]{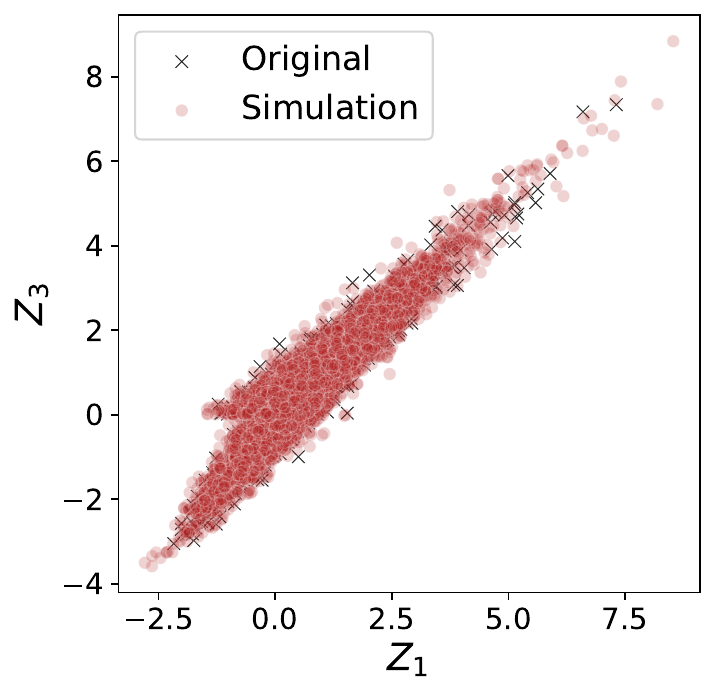}
&  \includegraphics[width=0.3\textwidth]{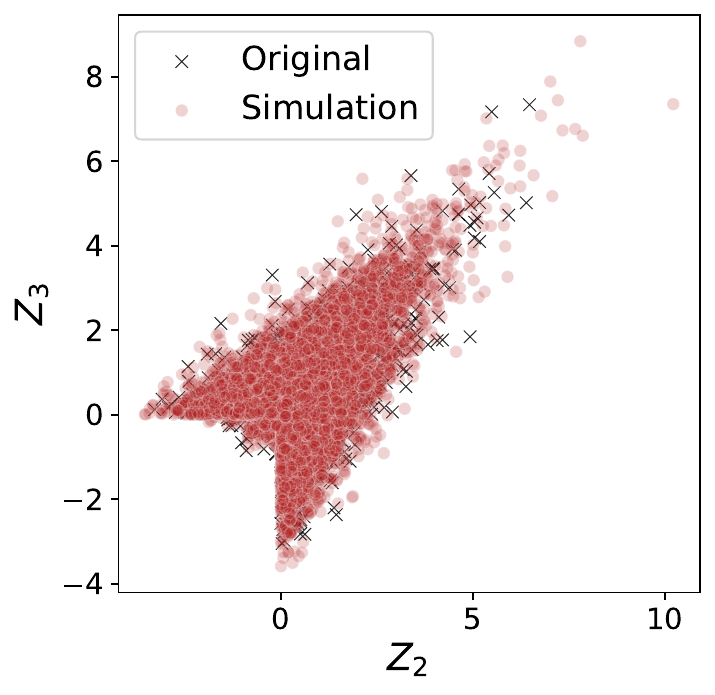}
\\
 a) & b) & b) \\ 
\end{tabular}
\caption{Bivariate representations of the  {simulated} sample of size 2,000 of a MGP vector $\bm Z \in \mathbb{R}^3$ with $\bm T$ (see Equation~\eqref{eq:SMGPRep}) distributed according to a centred multivariate Gaussian distribution with correlation coefficients $\rho_{1,2}=0.4,\rho_{1,3}=0.8$ and $\rho_{2,3}=0.1$ and the simulated sample $\bm Z^\ast$ of size $10000$. Figures a), b) and c) show scatter plots of the components  ($Z_1$, $Z_2$), ($Z_1$, $Z_3$) and ($Z_2$, $Z_3$) of the simulated sample (black crosses). On each plot, the components of the bootstrap sample ($Z_1^*$, $Z_2^*$), ($Z_1^*$, $Z_3^*$) and ($Z_2^*$, $Z_3^*$) are added (red circles).}
  \label{fig:Jointk3_scatter}
\end{figure}

}

 {Proposition \ref{prop:sim1} } guarantees that the samples obtained with  Algorithm~\ref{alg:JointMGP} \rev{are distributed according to the same distribution as the original observations.} 
\rev{The SB simulation procedure is applied to standard MGP vectors so that  the procedure itself ensures the preservation of the dependence structure of the original MGP vectors.   This is achieved by applying the bootstrap solely to the spectral component $\Delta$ in Equation \eqref{eq:MGPEq}, which encodes this dependence. Meanwhile, generating (new) independent exponential variables ensures that the SB procedure produces genuinely new observations.}

The proof can be found in Section~\ref{sec:theo} with the statement of the proposition. 


 {Figure~\ref{fig:Jointk3} presents the quantile-quantile (QQ) plots  between the simulated sample and bootstrap sample for each component.
The QQ plots show a good fit between the marginal distributions of the bootstrap samples $Z^\ast_j$ and the marginal distributions of the simulated sample $Z_j$ for $j \in \{1,2,3\}$. Consequently, the QQ plots indicate that Algorithm~\ref{alg:JointMGP} generates samples with the same marginal distributions as the ones of the simulated sample as shown in Proposition \ref{prop:sim1}.}

\begin{figure}
 \centering
  \begin{tabular}{ccc}\includegraphics[width=0.3\textwidth]{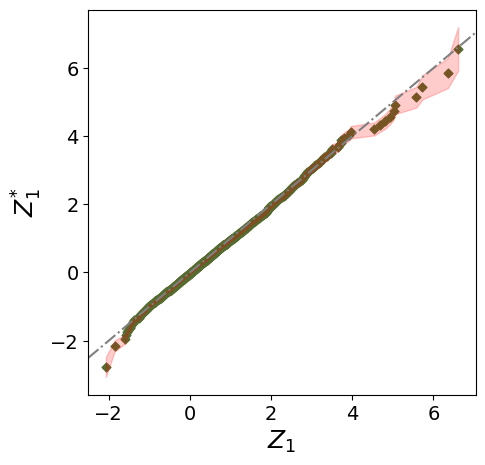} &  \includegraphics[width=0.3\textwidth]{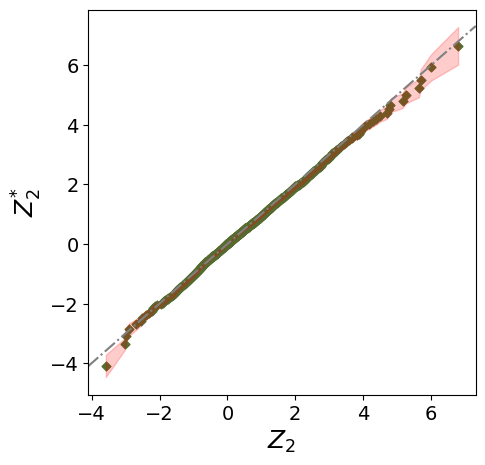} &
 \includegraphics[width=0.3\textwidth]{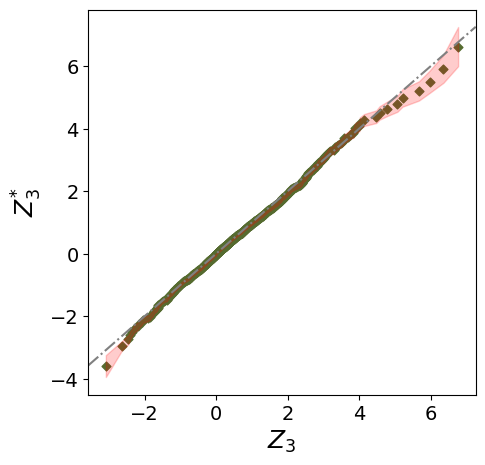} \\
 a) & b)&  c) 
 
\end{tabular}
\caption{Bivariate representations of the  {simulated} sample of size 2,000 of a MGP vector $\bm Z \in \mathbb{R}^3$ with $\bm T$ (see Equation~\eqref{eq:SMGPRep}) distributed according to a centred multivariate Gaussian distribution with correlation coefficients $\rho_{1,2}=0.4,\rho_{1,3}=0.8$ and $\rho_{2,3}=0.1$ and the  {bootstrap} sample $\bm Z^\ast$ of size $10000$. Figures a), b) and c) display the QQ plots of ($Z_1$, $Z^\ast_1$), ($Z_2$, $Z^\ast_2$) and ($Z_3$, $Z^\ast_3$), respectively, and their associated 95\% point-wise confidence intervals based on $1000$ bootstrap replications.}
  \label{fig:Jointk3}
\end{figure}

\section{Illustration: estimation of Tail Risk Metrics}\label{sec:illustation}

 {As mentioned in the introduction, the \rev{SB simulation procedure} addresses the problem of the scarcity of available observations in extreme regions, allowing for more reliable estimation of quantities depending on the tail of MGP distributions.}
  {We thus illustrate this by assessing the performance of the \rev{SB simulation procedure} (Algorithm \ref{alg:JointMGP} by estimating TRMs on both simulated data and on real financial data.

The three distinct TRMs that will be considered are introduced in Section~\ref{sec:TRM}. The simulation framework and the illustration of the \rev{SB simulation procedure} on simulated data are presented in Sections~\ref{subsec:NS}), and the illustration on real financial data is contained in Section~\ref{sec:DataReal}.}

\subsection{Tail Risk metrics}\label{sec:TRM}


  {As all TRMs are defined in terms of the Value-at-Risk (VaR), we first recall its definition. For a given level $\alpha \in (0,1)$, the VaR of level $1-\alpha$ of a random variable $X$}, denoted $\textrm{VaR}_{\alpha}(X)$, is defined as its $(1-\alpha)$-quantile: 
\[
\VaR_{\alpha}(X) = \inf\{x \in \mathbb{R} : \mathbb{P} \left(X > x\right)\leq \alpha\}. \, 
\]
  {In applications, one is interested in VaR for extremely large levels, i.e. $\alpha\approx 0$.}
  {The $\VaR_\alpha$ then corresponds to a high quantile of $X$. While the term \textit{Value-at-Risk} is standard in financial risk management \citep[e.g.][]{mcneil2015quantitative}, it is commonly referred to as the \textit{return level} in some other fields \citep[see][]{coles2001}.}

 Several limitations of the $\VaR$ have been identified in the literature. In particular, this TRM does not account for the severity of the random variable $X$, which may be a loss in financial applications \citep{yamai2005value}, or the magnitude of a climatic variable in environmental applications, e.g. the magnitude of daily precipitation \citep{grundemann2023extreme}. To address this shortcoming, alternative TRMs have been proposed, including the {\it Expected Shortfall} ($\ES$) \citep{artzner1999coherent}. The $\ES$  at level $1-\alpha$  {of X} is defined as 
 
\begin{equation}\label{eq:ES}
    \ES_{\alpha}(X) = \mathbb{E}\left[X \mid X>\VaR_{\alpha}(X)\right] 
\end{equation}
In words, the $\ES_\alpha$ corresponds to the  {expectation of observations larger than} $\VaR_\alpha$  \citep[see, e.g.,][]{artzner1999coherent,rockafellar2000optimization}. The $\ES$ is a useful tool for gaining insights into the severity of losses above a high quantile (i.e. the $\VaR$). It is therefore commonly used 
 in financial institutions for the calculation of the minimum capital requirements for market risk as specified by financial regulators \citep[see][for example]{basel2013}.  {In fields other than financial risk management, one may be interested in the $\ES$ rather than the $\VaR$.} Due to its importance, many techniques have been proposed to estimate the $\ES$ \citep[for a comprehensive review, see][]{nadarajah2014estimation}.   {In a parametric approach, the} $\ES$ is typically computed as the empirical mean   {of observations exceeding the $\VaR$, where the latter is estimated as a high quantile of a fitted extreme value distribution \citep[e.g.][]{coles2001}. Building on this framework, more sophisticated estimation techniques have been proposed \citep[][]{mcneil2000estimation, singh2013extreme}.}
 
  {Often, the severity of an event depends on several phenomena.} One limitation of the $\ES$ lies in  {the fact that it only considers the variable of interest $X$},  {ignoring the potential for asymptotic dependence between $X$ and other random variables that could impact the severity of $X$}. To address this, we consider two alternative risk metrics   {that account} for this asymptotic dependence   {structure} in their univariate risk evaluation of $X$. 

  {Let now $\bm{X} = (X_1, \dots, X_d)$ be a random vector in $\mathbb{R}^d$.  {The components of $\bm X$ represents variables that might have an impact of the severity of an phenomenon, they are classically called {\it risk factors}.}  A commonly used TRM in this multivariate setting is the marginal expected shortfall, defined} as the conditional expectation $\mathbb{E}\left[X_j \mid S(\bm{X}) \geq v\right]$, where $S$ is a chosen statistic of $\bm{X}$—such as the sum, minimum, or maximum—and $v \in \mathbb{R}$ is a threshold indicating the occurrence of an extreme event \citep{cai2015estimation}.
Building on this definition, we introduce a novel extension by conditioning on a joint extreme event involving all risk factors excluding the target variable $X_j$. Specifically, we define the {\it Multivariate Marginal Expected Shortfall} ($\MES_\alpha$) of $X_j$ at level $1-\alpha$ as follows:
\begin{equation}\label{eq:MES}
    \MES_{\alpha}(X_j;\bm X) = \mathbb{E}\left[X_j \mid \bm X_{-j}\geq \bm v^\alpha_{-j}\right] 
\end{equation}
where  {$\bm v^\alpha =\left(\VaR_\alpha(X_1), \ldots, \VaR_\alpha(X_d)\right) \in\mathbb{R}^{d}$ denotes the vector of component-wise Value-at-Risk at level $1-\alpha$ for $\bm X$}.  
Through this formulation the aim is to capture the behaviour of $X_j$ when similar risk factors reach extreme levels.


  {We further introduce a third TRM within the multivariate framework that captures the risk associated with the target variable $X_j$ when both $X_j$ and similar risk factors simultaneously} reach extreme levels. In the bivariate case   {($d = 2$), previous studies have explored} the incorporation of a dependent risk factor to enhance tail risk quantification for a target variable \citep[see, e.g.][]{Josaphat2021, goegebeur2024dependent}. In this work, we adopt the concept of the {\it Dependent Conditional Tail Expectation} ($\DCTE$) as defined by \citet{goegebeur2024dependent} in the bivariate setting, and extend it to the general multivariate case as follows:
\begin{equation}\label{eq:DCTE}
    \DCTE_{\alpha}(X_j; \bm X) = \mathbb{E}\left[X_j | \bm X\geq \bm v^\alpha \right]
\end{equation}
 


  {In the following sections, we apply the \rev{SB simulation procedure} developed in Section~\ref{sec:joint} on both simulated and real data to enhance the empirical estimation of the three TRMs under consideration: the expected shortfall (ES; Eq.~\eqref{eq:ES}), the multivariate marginal expected shortfall (MMES; Eq.~\eqref{eq:MES}) and the dependent conditional tail expectation (DCTE; Eq.~\eqref{eq:DCTE}).}


\subsection{Illustration on simulated data}\label{subsec:NS}

In this section, we first present the simulation framework and  then the illustration of Algorithm~\ref{alg:JointMGP} on simulated data.)

  {The framework is inspired by the field of finance and more specifically financial returns: we consider a 3-dimensional random vector $\bm X =(X_1,X_2,X_3)$ with  marginal Student's $t$-distribution with relatively low degrees of freedom $\nu_1, \nu_2, \nu_3$, respectively. The Student's $t$ distribution is heavy-tailed and the tail index is equal to $1/\nu$. So, the smaller $\nu$ the heavier the tail. In the following, we will consider that the marginal degrees of freedom will be larger than 1 so that the expectation of the components is finite.  This choice of distribution for the marginals is chosen  so as to mimic financial returns, which are typically heavy-tailed.}

To fully characterize the joint distribution of $\bm X$, we need to define the  dependence structure in addition to the marginal distributions. Sklar's theorem \citep{sklar1959fonctions} states that   {for any} joint distribution $\bm F$ of a random vector $\bm X$   {in $\mathbb{R}^d$} with   {continuous} marginals $F_1,\dots,F_d$, there exists a unique copula $\mathcal{C}:[0,1]^d\to [0,1]$ such that $\bm F(x_1,\dots,x_d) = \mathcal{C}\left(F_1(x_1),\dots,F_d(x_d)\right)$. Copulas   {thus provide a} powerful framework for modeling the dependence structure of a random vector separately from its marginals. 

Let us also recall that an underlying assumption of the   {SB procedure} is that the components of $\bm X$ have non-trivial asymptotic tail dependence. In order to ensure that this is indeed the case, the Gumbel copula \citep{nelsen2006introduction} is employed to derive dependent extremes in the upper tail.

\begin{align}\label{eq:copGumb}
    \mathcal{C}(\bm y) := \exp\left(- \left( \sum_{i=1}^3\left[-\log(y_i)\right]^{\theta
}\right)^{1/\theta}\right),
\end{align}
where $\theta \geq 1$ is the copula parameter. The larger $\theta$, the stronger the asymptotic dependence structure between the components of $\bm X$. Consequently, the pair $\left(\bm\nu,\theta\right)$ parametrize the structure of dependence through Equation~\eqref{eq:copGumb}, fully characterizes the joint distribution of $\bm X$.

The numerical experiments are performed on simulated data sets $\mathcal D\in \mathbb{R}^{1500\times 3}$, with $\nu_1 = 2$,  $\nu_2=3$, $\nu_3=2.5$. From this parametric framework, we can derive the theoretical values of the TRMs and infer quantities describing the tail of some conditional distribution, that will be used as benchmarks in the following sections.

  {Algorithm~\ref{alg:JointMGP} enables the generation of new realisations from  {a standard MGP distribution}. However, in practice, the observed data are not necessarily  {standard} MGP-distributed. To address this, we rely on the classical approximation from EVT: after appropriate marginal standardization, one may assume that, for a sufficiently high threshold vector $\bm u$, the distribution of the multivariate exceedances $\bm X - \bm u \mid \bm X \nleq \bm u$ can be approximated by a standard MGP distribution,  {see Section~\ref{sec:MGP} or \citep{coles2001}}. More precisely, we consider the following two steps: 
\begin{enumerate}
  \item \textbf{Standardization to exponential margins:} transform the component of $\bm X$ to have standard exponential margins using the probability integral transform. For each $j=1,\ldots,d$, define  $X^E_j = \log \left(1- F_{j}(X_j)\right),$
  where $F_{j}$ is the $j$-th marginal of $\bm X$. 
  \item \textbf{Computation of multivariate exceedances:} following the definition of \cite{rootzen2006multivariate}, compute the multivariate excesses vector as follows \begin{equation}\label{eq:Z_XE}
      \bm Z = \bm X^E - \bm u^E \mid \bm X^E  \not \leq \bm u^E \, ,
  \end{equation} 
where $\bm u^E$ is a suitably chosen threshold on the standard scale, see right below. 
\end{enumerate}}

  {As recalled in Section \ref{sec:MGP}, if the threshold $\bm u^E$ is sufficiently high, the distribution of the vector of excesses $\bm Z$ defined in \eqref{eq:Z_XE} can be approximated by a MGP distribution. This approximation justifies the application of the \rev{SB simulation procedure} in order to increase the number of joint extreme observations. The choice of $\bm u^E$ is a classical discussion in EVT. One approach that can be considered is the use of graphical tools such as the stability plot approach proposed in \cite{kiriliouk2019peaks}. }

  {After running the algorithm on the transformed data $\bm{Z}$}, the resulting bootstrap sample outputs $\bm Z^\ast$ can be projected back into the original space  {(e.g. of the simulation data)} using the inverse transformation:
\[
X^\ast_j = F_j^{-1}\left(1-e^{Z^\ast_j+u^E_j}\right), \mbox{ for } j=1,\ldots,d
\] 
where $F^{-1}_j$ denotes the quantile function of the $j$-th  {component} $X_j$. 
  {This backward transformation enables estimation of TRMs, or any other quantity of interest, on the  {simulated sample} $\bm X$ to be performed using the  {bootstrap} samples $\bm X^\ast$.}


Finally, \rev{in the sequel, we will consider} three different samples:
\begin{enumerate}
    \item {\bf  {Simulated} sample}  generated using the parametric framework above, denoted $\mathcal{D}$. 
    \item {\bf  {Bootstrap} sample}  obtained using Algorithm~\ref{alg:JointMGP}, denoted $\mathcal{D^\ast}$.
    \item {\bf Extended sample}  corresponding to the concatenation of the  {simulated} and  {bootstrap} samples, denoted $\mathcal{D} \cup \mathcal{D^\ast}$.
\end{enumerate}


We assess the performance of \rev{the SB simulation } procedure (Algorithm~\ref{alg:JointMGP}) \rev{on the estimation of TRMs}, as defined in Equations~\eqref{eq:ES}, \eqref{eq:MES}, and \eqref{eq:DCTE}, across various parameter settings. Special emphasis is placed on the strength of asymptotic dependence, which is governed by the copula parameter $\theta$ and the level $\alpha$ (see Supplementary Material, Section 1). Estimated TRMs are compared to their respective theoretical benchmark values (see Supplementary Material, Section 2). For simplicity, we focus on the first component $X_1$ being our target risk factor, though the procedure applies equally to $X_2$ and $X_3$.


Since all TRMs are derived from the $\VaR$, particular care must be taken to ensure its accurate estimation. However, this aspect is not the main focus of our current investigation. In this study, the performance of our approach is assessed using the theoretical $\VaR$, specifically computed as a high quantile of a Student’s $t$-distribution. Alternative methods for estimating the $\VaR$—including an empirical approach and estimation based on a univariate extreme value distribution—have also been explored. Detailed results of these analyses are provided in the Supplementary Material (Section 4). Notably, they suggest that when the theoretical $\VaR$ is unavailable, as is often the case in real-world applications, the univariate extreme value approach offers a promising alternative for estimating the $\VaR$.

\rev{Next, we fix the dependence parameter $\theta$ to be equal to $2.6$ (this value is driven by typical financial applications, other values of copula parameter have been tested, yielding to similar improved estimation results).}
Given that the $\VaR$ estimation technique and \rev{$\theta$} are fixed, we can now proceed to investigate the performance of our approach with respect to \rev{the level} $\alpha$ and \rev{the bootstrap sample size $m$}. 
Since estimation errors may be sensitive to the specific  {simulated} sample $\mathcal{D}$, it is important to ensure that the observed performance of the approach is not tied to a particular realization of $\mathcal{D}$. To this end, \rev{100} distinct  {simulated} samples $\mathcal{D}$ were generated, each of size $n$ and drawn from the same joint distribution described in Section~\ref{subsec:NS}. For each  {simulated} sample, a set of  \rev{100} {bootstrap} samples $\mathcal{D}^\ast$, each of size \rev{$m\in\{100,1000,10000\}$}, was then produced using the  \rev{SB simulation} procedure outlined in Algorithm~\ref{alg:JointMGP}. TRMs were estimated \rev{on each of the 100 simulated samples $\mathcal D$ and of extended samples $\mathcal D \cup \mathcal D^{\ast}$}.

 \rev{Figure~\ref{fig:AlphaThetaTRM} and Table~\ref{tab:ImpactThetaTRM}  summarize the results for TRMs estimations at two different levels  {$\alpha \in \{0.0025,0.0003\}$}. 
Boxplots of the distributions of the TRMs estimates (Figure~\ref{fig:AlphaThetaTRM}) displays the TRM estimates based on the original samples $\mathcal D$ (grey hatched box) and on the extended samples $\mathcal D \cup \mathcal D^\ast$ (red dotted boxes) for $m\in\{100,1000,10000\}$. They show}\begin{figure}
\begin{tabular}{cc}
    \includegraphics[width= 0.47\textwidth]{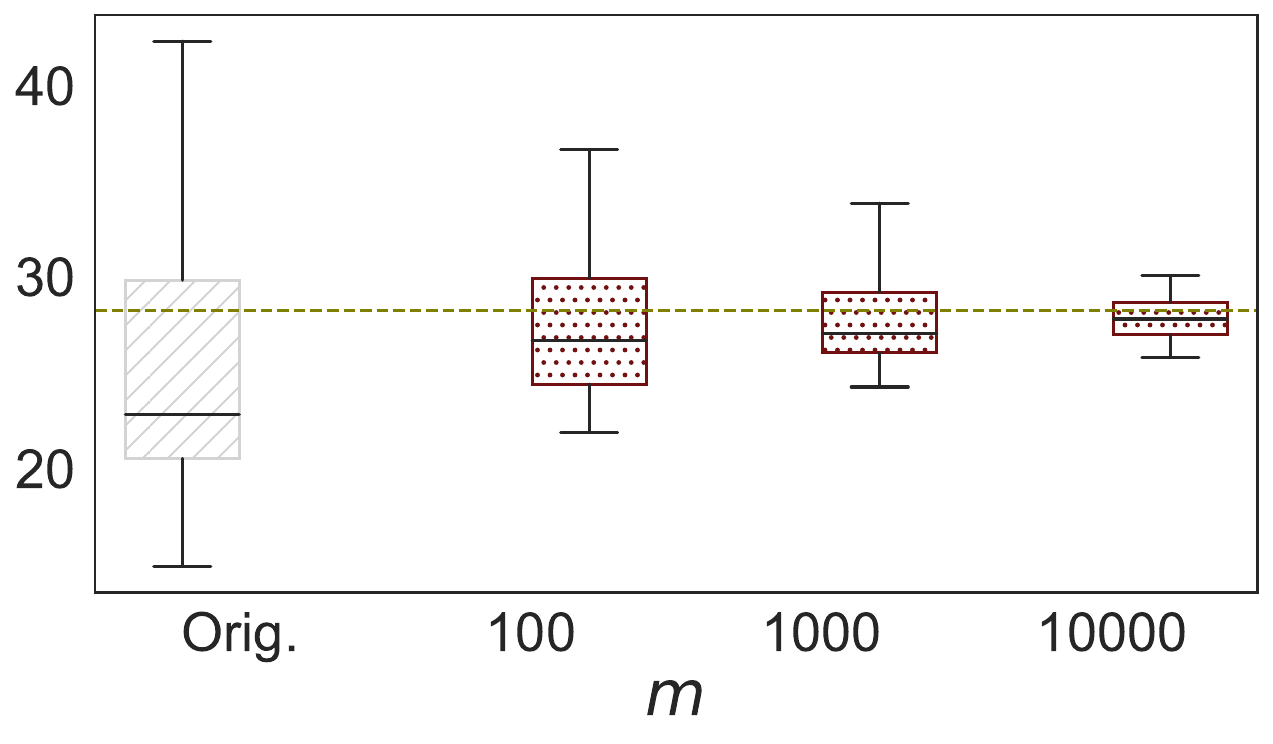} &   \includegraphics[width=0.47\textwidth]{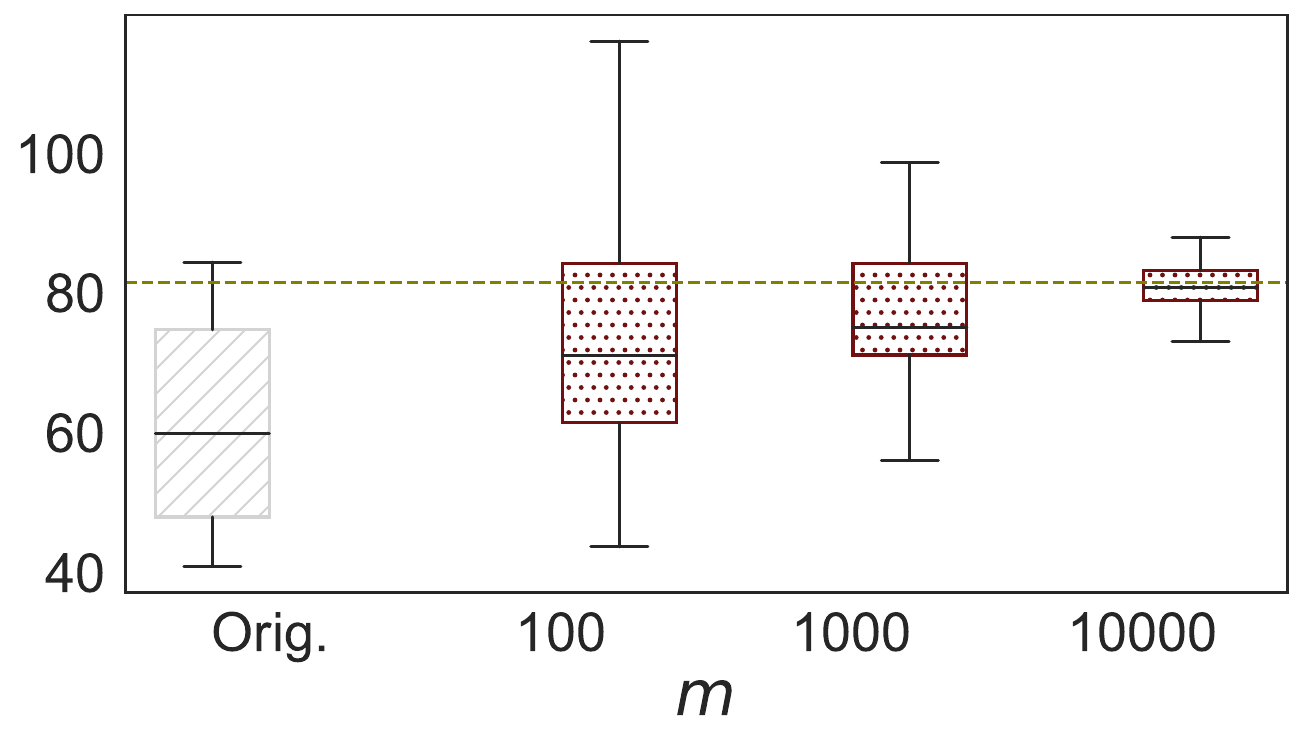} \\~\\
    a) $\ES_{ {0.0025}}$ & b) $\ES_{ {0.0003}}$ \\~\\
 
    \includegraphics[width=0.47\textwidth]{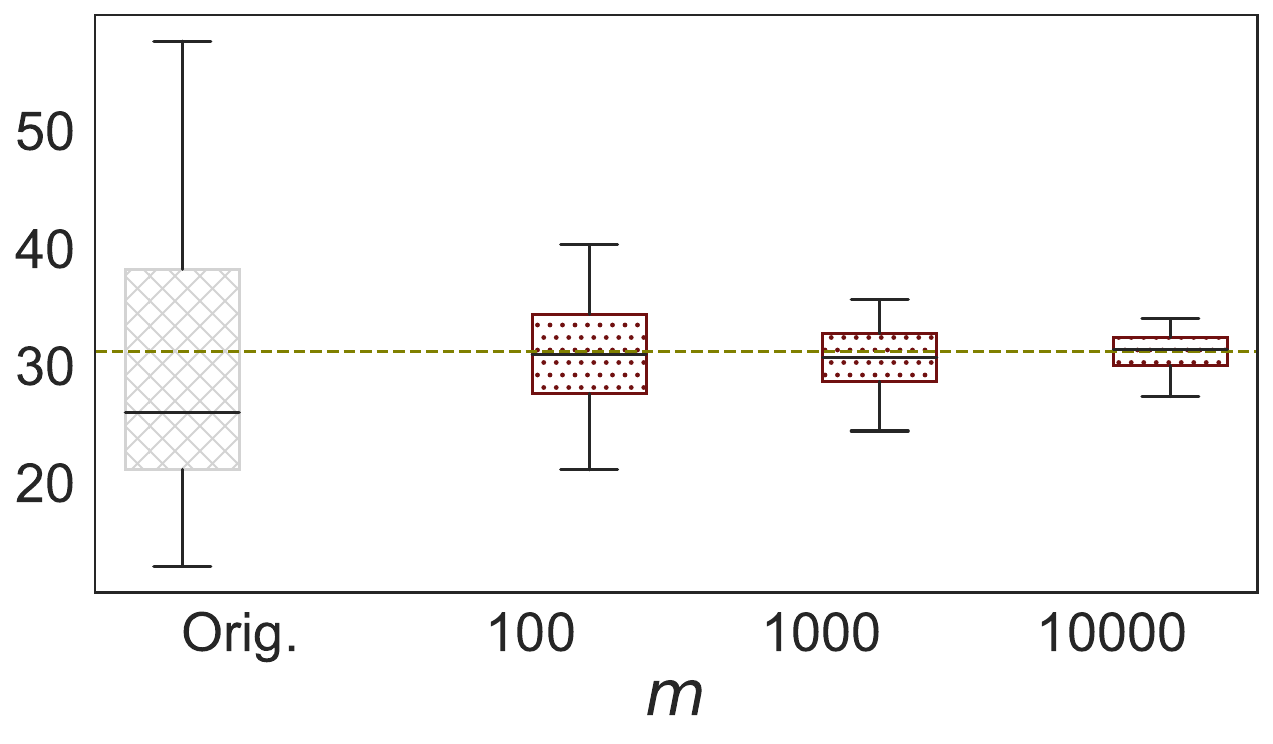} &    \includegraphics[width=0.47\textwidth]{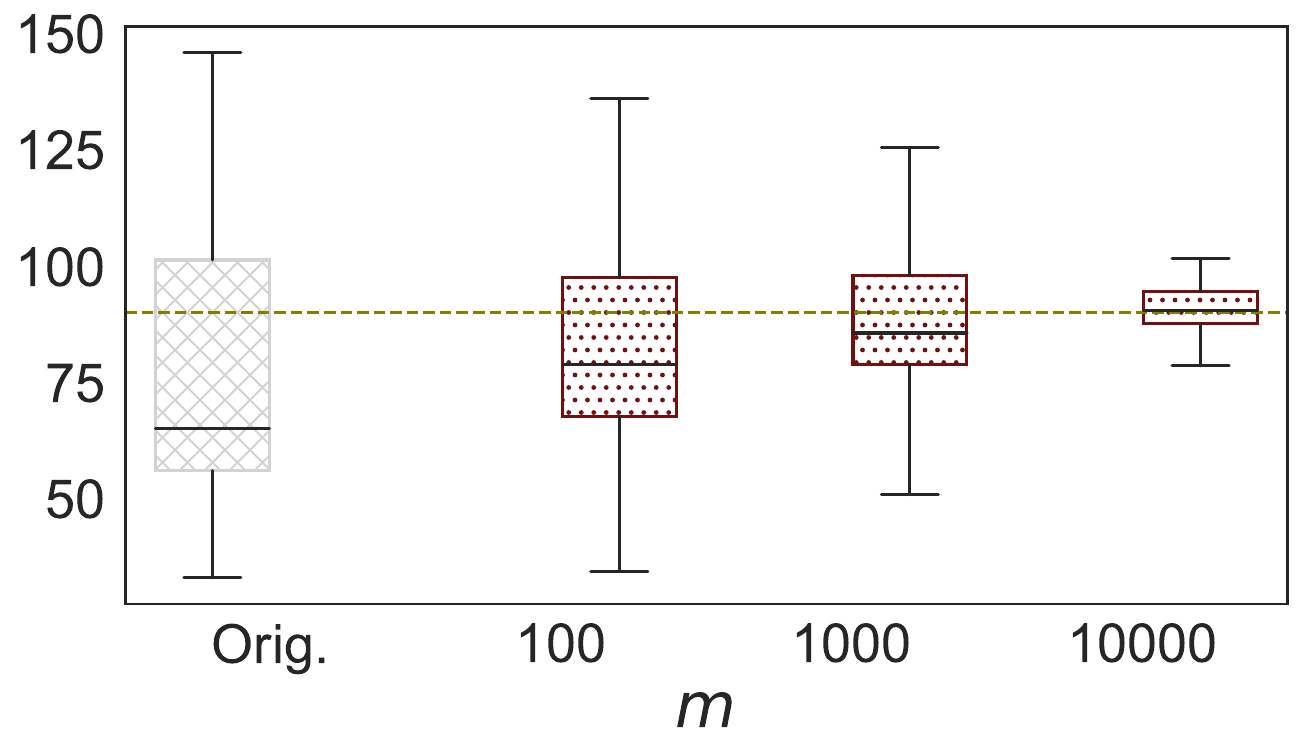}\\~\\
    c) $\MES_{ {0.0025}}$ & d) $\MES_{ {0.0003}}$ \\~\\ 

    \includegraphics[width=0.47\textwidth]{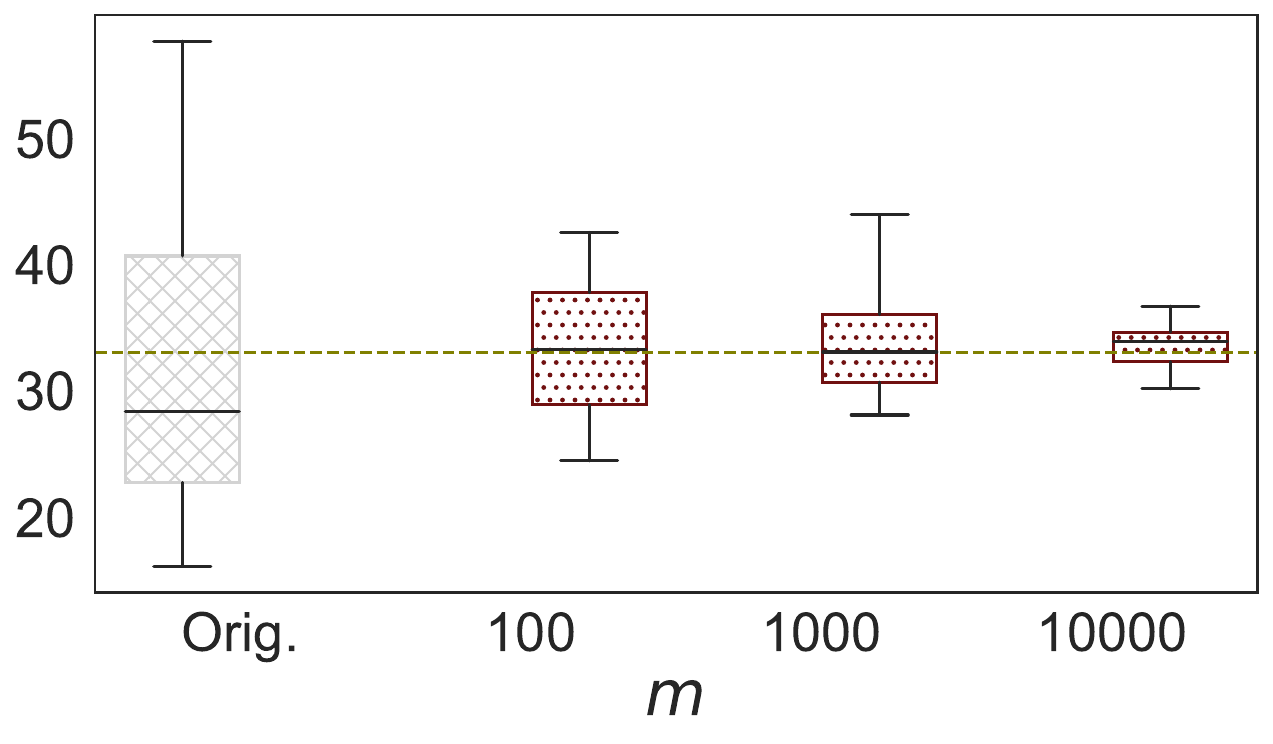} &  
     \includegraphics[width=0.47\textwidth]{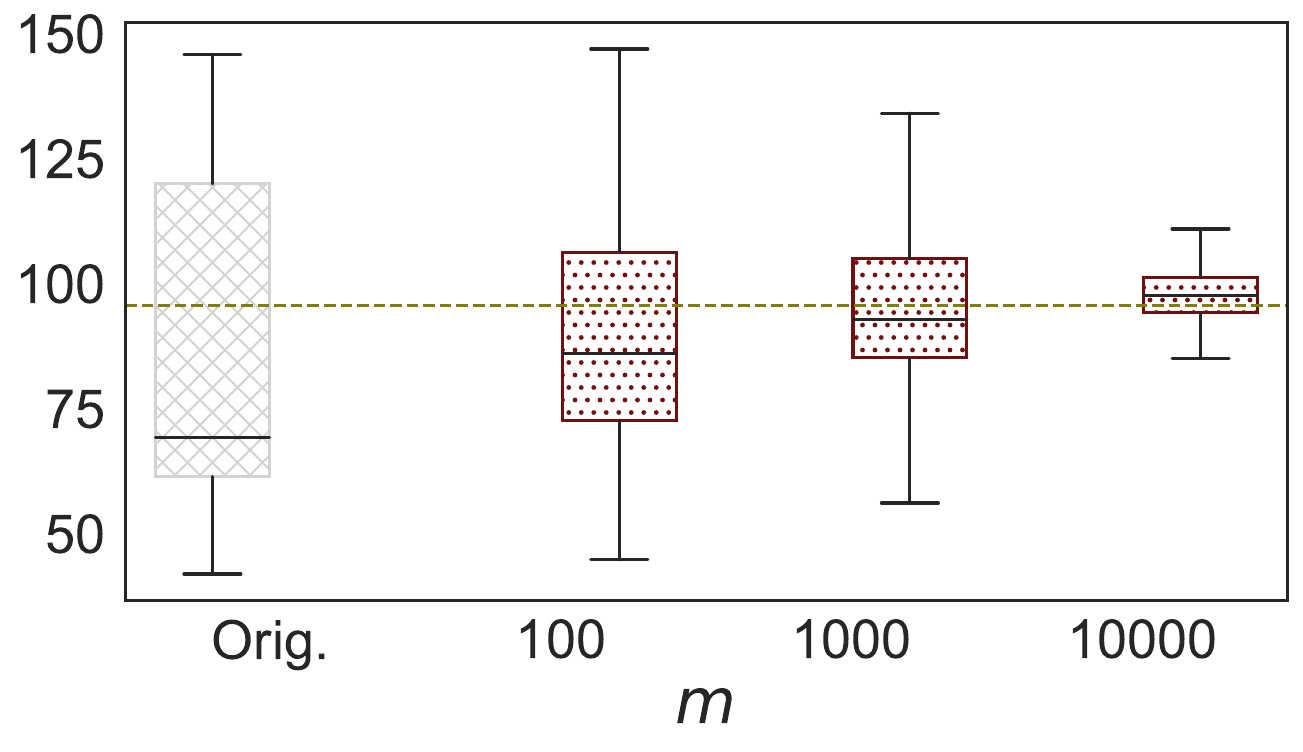}\\~\\
    e) $\DCTE_{ {0.0025}}$ & f) $\DCTE_{ {0.0003}}$\\
    \end{tabular}
  \caption{\rev{Distribution of TRM estimates on 100  {simulated} samples $\mathcal{D}$ {\it (grey hatched box)} and on 100 extended samples \rev{$\mathcal{D} \cup \mathcal{D^\ast}$} {\it (red dotted boxes)}, with $m\in\{100,1000,10000\}$ for the $\ES_\alpha$ (Figures a), b)), the $\MES_\alpha$ (Figures c), d)) and the $\DCTE_\alpha$ (Figures e), f)) at levels {$\alpha \in \{0.0025,0.0003\}$} and $\theta= 2.6$. On each graph, the dashed green horizontal line corresponds to the associated theoretical value.}}
  \label{fig:AlphaThetaTRM}
\end{figure}
\rev{ that the TRMs \rev{estimates on the extended samples $\mathcal{D} \cup \mathcal{D^\ast}$ are more accurate and exhibit less dispersion than those based on the simulated samples $\mathcal{D}$. This is to be expected, as the SB simulation procedure enables us to perform empirical estimation on a larger sample of extremes.}
One may also observe that the estimation on $\mathcal{D} \cup \mathcal{D^\ast}$ induce fairly close  {estimates}  regardless of the level $\alpha$. Note that the joint simulation of multivariate extremes using Algorithm~\ref{alg:JointMGP} improves not only the estimation of $\MES$ and $\DCTE$ (defined with the joint distribution of $\bm X$) but also the estimation of $\ES$ even though it only evolves the marginal distribution of $X_1$.} 

\begin{table}
\begin{center}
\begin{tabular}{c}
\begin{tabular}{c|rr|rr|rr}
\multirow{2}{*}{$\alpha$} & \multicolumn{2}{c|}{$\ES_\alpha$}   & \multicolumn{2}{c|}{$\MES_\alpha$}    & \multicolumn{2}{c}{$\DCTE_\alpha$}  \\
\cline{2-7} & \multicolumn{1}{c|}{$\mathcal{D}$} & \multicolumn{1}{c|}{$\mathcal{D} \cup \mathcal{D^\ast}$} & \multicolumn{1}{c|}{$\mathcal{D}$}   & \multicolumn{1}{c|}{$\mathcal{D} \cup \mathcal{D^\ast}$} & \multicolumn{1}{c|}{$\mathcal{D}$}    & \multicolumn{1}{c}{$\mathcal{D} \cup \mathcal{D^\ast}$} \\ 
 \hline

0.0025     & \multicolumn{1}{r|}{3.8 (2.2)} & 114.7 (6.7)   & \multicolumn{1}{r|}{2.6 (1.6)}  & 83.5 (4.6)                & \multicolumn{1}{r|}{2.3 (1.6)} & 74.4 (3.9) \\ 
\hline
0.0003   & \multicolumn{1}{r|}{0.4 (0.6)} & 13.7 (0.9)  & \multicolumn{1}{r|}{0.3 (0.5)} & 10 (0.7)  & \multicolumn{1}{r|}{0.2 (0.5)} & 8.9 (0.6)  
\end{tabular}\\~\\
\end{tabular}
\caption{\rev{Average the number of observations (standard deviations) on which the TRMs are estimated for the  {simulated} samples $\mathcal{D}$ and for the  {extended} samples $\mathcal{D} \cup \mathcal{D^\ast}$, with $m= 10000$,  for  {$\alpha \in \{0.0025,0.0003\}$} and  $\theta=2.6$.}}
 \label{tab:ImpactThetaTRM}
 \end{center}
\end{table}

\rev{The accuracy of the TRM estimates stems from the use of a larger estimation sample relative to the simulated sample size}. \rev{Table~\ref{tab:ImpactThetaTRM} shows the average number of observations on which the three TRMs are estimated when considering the  {simulated} samples $\mathcal D$ and  {bootstrap} samples $\mathcal D^\ast$ for {$\alpha \in \{0.0025,0.0003\}$}. Overall, it can be noted that for the considered levels of $\alpha$, the TRMs estimation on the  {simulated} samples is conducted with limited data, that is, here with at most 4 observations.} In contrast,  \rev{the SB simulation procedure} allows for a larger sample size to estimate the various TRMs, extending the  {simulated} observations to 100 observations in some cases. 

\subsection{Illustration on real financial data}\label{sec:DataReal}

This section presents a more practical illustration of the SB  {procedure}, and TRMs estimation, using real financial data. The data set consists of weekly negative returns on  stocks from three UK banks: HSBC, Lloyds and RBS. The corresponding time series are denoted $X_1$, $X_2$ and $X_3$ respectively. The sample size of each time series is $n=470$, with temporal resolution spanning from 29/10/2007 to 17/10/2016. The data used in this study were initially extracted from Yahoo Finance, and considered in \citep{kiriliouk2019peaks}. 

As in Section~\ref{subsec:NS}, Algorithm~\ref{alg:JointMGP} is applied  in order to estimate TRMs  with this real data set $\bm X =(X_1,X_2,X_3)$. For this, and as done with the simulated data, the first step is to standardise the data to common exponential margins (see Section~\ref{subsec:NS}).  In our case, we have chosen to fit a Student's $t$-distribution on each risk factor, which is often used for heavy-tailed financial data  \citep[see e.g.][]{mcneil2015quantitative}. 
Goodness-of-fit plots can be found in the Supplementary Material (see Section 5) along with the estimated parameters of the fitted Student's $t$-distributions.


A set of $100$  {bootstrap} samples of weekly negative returns of size $m=10000$ were generated using  {the \rev{SB simulation procedure} (Algorithm~\ref{alg:JointMGP})} in order to  {estimate} the three TRMs  {introduced in Section~\ref{sec:TRM}}. The results are displayed in Table~\ref{tab:TRMRealData} along with the estimations on the original sample $\mathcal{D}$ , i.e.  {the real data without increasing the amount of available data}. The performance assessment conducted under the parametric framework has shown that TRM estimates on the  {bootstrap} samples $\mathcal{D^\ast}$  were more precise than estimations on the extended sample $\mathcal{D} \cup \mathcal{D^\ast}$, hence, we exclusively discuss results from the  {bootstrap} samples $\mathcal{D^\ast}$ in this section. 
 Table~\ref{tab:TRMRealData} shows that the  {SB procedure} allows for the estimation of the TRMs that would be unfeasible if only the original sample was considered. This is to be expected, given that the objective of the  {SB procedure} is to increase the number of available observations in extreme regions in order to enable the performance of empirical estimations in a reliable manner.

\begin{table}
    \centering
   \begin{tabular}{c|rr|rr|rr}
\multirow{2}{*}{Stock} & \multicolumn{2}{c|}{$\ES_{0.0025}$}  & \multicolumn{2}{c|}{$\MES_{0.0025}$}          & \multicolumn{2}{c}{$\DCTE_{0.0025}$}   \\
\cline{2-7} & \multicolumn{1}{c|}{$\mathcal{D}$}  & \multicolumn{1}{c|}{$\mathcal{D^\ast}$} & \multicolumn{1}{c|}{$\mathcal{D}$}        & \multicolumn{1}{c|}{$\mathcal{D^\ast}$} & \multicolumn{1}{c|}{$\mathcal{D}$}    & \multicolumn{1}{c}{$\mathcal{D^\ast}$} \\
\hline

HSBC & \multicolumn{1}{r|}{0.24 (-) } & 0.28 (0.014)    & \multicolumn{1}{c|}{NA} &0.29 (0.023)   & \multicolumn{1}{r|}{0.11 (-)} & 0.32 (0.026)     \\ 
 LL  & \multicolumn{1}{r|}{0.55 (-)} & 0.83 (0.092) & \multicolumn{1}{c|}{NA} & 0.92 (0.139)   & \multicolumn{1}{c|}{NA}  & 1.04 (0.168)   \\ 
 RBS  & \multicolumn{1}{r|}{0.62 (-)} & 0.56 (0.031)   & \multicolumn{1}{c|}{NA} & 0.62 (0.052)      & \multicolumn{1}{c|}{NA} &0.66 (0.057)              
\end{tabular}
\caption{Empirical mean estimates, and standard deviation in brackets, of the TRMs for the three negative returns of interest, using the original data $\mathcal{D}$ and using  {Algorithm~\ref{alg:JointMGP}} with $100$  {bootstrap} samples $\mathcal{D^\ast}$ of size $m=10000$ each . For each TRM, a threshold level  {$\alpha=0.0025$} is considered. NA meaning that there are no observations above the theoretical $\VaR$, so that the computation of the TRM was not possible.}
\label{tab:TRMRealData}
\end{table}

\section{  {Discussion}}

  {This work focused on the development of a non-parametric bootstrap procedure (SB procedure) in order to enlarge the number of joint extreme observations in a data sample. The algorithm is supported by the spectral representation of multivariate GP distributed random vectors. We applied this procedure to the estimation of various risk metrics that can be of great importance for insurance companies.

We want to emphasize that this data augmentation procedure is not restricted to such applications. As highlighted in the summary paper \citep{elmethni2024}, the \rev{SB simulation procedure} can serve as an initial step in a wide range of estimation problems in EVT, including those arising in extremal regression, providing more accurate estimators without relying on too restrictive assumptions \citep[see e.g.,][]{davison2023tail}. Another possible direction for future work is to combine the non-parametric \rev{SB simulation procedure} with generative models, such as Generative Adversarial Networks \citep{Goodfellow2020}, in order to artificially expand the number of extreme observations in the training dataset.}

\paragraph{Codes} All the codes and data are publicly available at \url{https://github.com/MadharNisrine/MultivariateExtremeSimulator.git}.

\section{Theoretical support}\label{sec:theo}

   {Theoretical support for Algorithm~\ref{alg:JointMGP} is provided below through Proposition~\ref{prop:sim1}, which establishes that conditionally on the original sample $\bm \Delta_1, \dots, \bm \Delta_n$, each simulated vector $\bm Z_l^\ast$ (for fixed $l = 1, \dots, m$) can be well approximated by a MGP distribution, provided that the original sample size is sufficiently large, i.e. convergence in distribution as $n \to \infty$.} 





 \begin{proposition}
\label{prop:sim1}
        Let $\bm\Delta_1,\dots,\bm\Delta_n$ be $n$ i.i.d.  realisation of $\bm\Delta$ defined in Equation~\eqref{eq:wk_w1}, and let $\bm\Delta_1^\ast,\dots,\bm\Delta_m^\ast$ be $m$ bootstrap samples of $\bm\Delta_1,\dots,\bm\Delta_n$. Define the associated i.i.d. \textit{SB} samples $$\bm Z_l^\ast=\tilde{E}_l+\bm\Delta_l^\ast,\quad 1\leq l\leq m,$$
        where $\tilde{E}_l\sim Exp(1)$ and is independent of $\bm\Delta_l^\ast$.
        Then, for all $l=1,\dots,m,$ and conditionally on $\bm\Delta_1,\dots,\bm\Delta_n,$ $\bm Z_l^\ast$ converges in distribution to \rev{the MGP distributed vector $\bm Z=E + \bm\Delta$ as $n\to\infty$}. 
    \end{proposition}
    \begin{proof}
        Let $\bm x =(x_1,\dots,x_d)\in\mathbb{R}^d$ and $l=1,\dots,m$. \begin{eqnarray*}
            \mathbb{P}\left(\bm Z_l^\ast\leq \bm x\mid \bm \Delta_1,\dots,\bm\Delta_n\right) &=& \mathbb{P}\left(\tilde{E}_l\leq \bm x - \bm\Delta_l^\ast\mid \bm \Delta_1,\dots,\bm\Delta_n\right), 
            \end{eqnarray*}
            where the right-hand side inequality is meant component-wise, \rev{thus}
            $$\tilde{E}_l\leq \bm x - \bm\Delta_l^\ast \Longleftrightarrow  \tilde{E}_l \leq \rev{\min_j(} x_j-\Delta_{l,j}^\ast\rev{)}.$$
            \rev{\begin{eqnarray*}
                \mathbb{P}\left(\bm Z_l^\ast\leq \bm x\mid \bm \Delta_1,\dots,\bm\Delta_n\right) &=& \mathbb{E}[\mathbbm{1}_{\bm Z_l^\ast\leq \bm x} \mid \bm \Delta_1,\dots,\bm\Delta_n] \\
                &=& \mathbb{E}[\mathbb{E}[\mathbbm{1}_{Z_l^\ast\leq \bm x} \mid \bm \Delta_l^\ast, \bm \Delta_1,\dots,\bm\Delta_n]\mid \bm \Delta_1,\dots,\bm\Delta_n]. 
                \end{eqnarray*}
                On the other hand, and since $\tilde E_l$ is independent of $\bm \Delta_l^\ast, \bm \Delta_1,\dots,\bm\Delta_n$, \begin{eqnarray*}
                    \mathbb{E}[\mathbbm{1}_{Z_l^\ast\leq \bm x} \mid \bm \Delta_l^\ast, \bm \Delta_1,\dots,\bm\Delta_n] &=& 
                    \mathbb{E}[\mathbbm{1}_{\tilde E_l\leq \min_j( x_j-\Delta_{l,j}^\ast)} \mid \bm \Delta_l^\ast, \bm \Delta_1,\dots,\bm\Delta_n]
                    \\ &=& 
                    \int_0^{+\infty} \mathbbm{1}_{t\leq \min_j(x_j-\Delta_{l,j}^\ast)} e^{-t}dt \\
                    &=& \min(0, 1-e^{-\min_j(x_j-\Delta_{l,j}^\ast)}).
                \end{eqnarray*}
                Therefore, } 
                \begin{equation}
                    \label{eq::cdf_boot}
            \mathbb{P}\left(\bm Z_l^\ast\leq \bm x\mid \bm \Delta_1,\dots,\bm\Delta_n\right) = 1 -\mathbb{E}\left[ \min\left(1,e^{-\min_{1\leq j\leq d} \{x_j-\Delta^\ast_{l,j}\}}\right)\mid \bm \Delta_1,\dots,\bm\Delta_n\right] \end{equation}
            \rev{By construction of the bootstrap sample, the conditional distribution of $\bm \Delta^\ast_l$ given $\bm\Delta_1,\dots,\bm \Delta_n$ is the empirical distribution of $\bm\Delta_1,\dots,\bm \Delta_n$:
            $$\forall \bm x\in\mathbb{R}^d,\mbox{ } \mathbb{P}\left(\bm\Delta_l^\ast\leq \bm x\mid \bm\Delta_1,\dots,\bm \Delta_n\right)=\frac{1}{n}\sum_{i=1}^{n}\mathbbm{1}_{\bm\Delta_i\leq \bm x}.$$ Then, by Glivenko-Cantelli theorem, the empirical cdf on the right hand side converges in probability to the common cdf of the $\bm\Delta_i$'s. It follows that $\bm \Delta^\ast_l$ given $\bm\Delta_1,\dots,\bm \Delta_n$ converges in distribution to $\bm\Delta$, thus for any bounded and measurable test function $g:\mathbb{R}^d\to\mathbb{R}$ $$\mathbb{E}\left[g(\bm\Delta_l^\ast)\mid \bm\Delta_1,\dots,\bm \Delta_n\right]\to \mathbb{E}\left[g(\bm\Delta)\right],\text{ as } n\to\infty.$$ For a fixed $\bm x\in\mathbb{R}^d$, the function $f$ defined on $\mathbb{R}^d$ by $f(\bm\delta)=\min\left(1,e^{-\min_{1\leq j\leq d} \{x_j-\delta_{j}\}}\right)$ is continuous and bounded (by $1$), so as $n\to\infty$ $$ \mathbb{E}\left[ \min\left(1,e^{-\min_{1\leq j\leq d} \{x_j-\Delta^\ast_{l,j}\}}\right)\mid \bm \Delta_1,\dots,\bm\Delta_n\right] \to \mathbb{E}\left[ \min\left(1,e^{-\min_{1\leq j\leq d} \{x_j-\Delta_{j}\}}\right)\right].$$ The limit on the right-hand side is the cumulative distribution function of the MGP distributed vector $\bm Z=E+\bm\Delta$ \citep[see][Proposition 8]{rootzen2018multivariate2}, and combining with \eqref{eq::cdf_boot} we get the statement of the proposition.}
    \end{proof}
\rev{
\paragraph{Acknowledgements} 
We thank the referees for their comments which
have led to important improvements of the paper.}

\bibliographystyle{chicago}%
\bibliography{mybibsimu}

@unpublished{elmethni2024,
  TITLE = {{Four contemporary problems in extreme value analysis}},
  AUTHOR = {El Methni, Jonathan and Girard, St{\'e}phane and Legrand, Juliette and Stupfler, Gilles and Usseglio-Carleve, Antoine},
  URL = {https://hal.science/hal-04780678},
  NOTE = {working paper or preprint},
  YEAR = {2024},
  HAL_ID = {hal-04780678},
  HAL_VERSION = {v1},
}

@article{Goodfellow2020,
author = {Goodfellow, Ian and Pouget-Abadie, Jean and Mirza, Mehdi and Xu, Bing and Warde-Farley, David and Ozair, Sherjil and Courville, Aaron and Bengio, Yoshua},
title = {Generative adversarial networks},
year = {2020},
issue_date = {November 2020},
publisher = {Association for Computing Machinery},
address = {New York, NY, USA},
volume = {63},
number = {11},
issn = {0001-0782},
url = {https://doi.org/10.1145/3422622},
doi = {10.1145/3422622},
journal = {Commun. ACM},
month = oct,
pages = {139–144},
numpages = {6}
}

@book{beirlant2004,
	author = {Beirlant, J. and Goegebeur, Y. and Segers, J. and Teugels, J.L.},
	isbn = {9780470012376},
	lccn = {2004051046},
	publisher = {Wiley},
	series = {Wiley Series in Probability and Statistics},
	title = {Statistics of Extremes: Theory and Applications},
	url = {https://books.google.fr/books?id=jqmRwfG6aloC},
	year = {2004}
}

@article{frechet1927loi,
  title={Sur la loi de probabilit{\'e} de l'{\'e}cart maximum},
  author={Fr{\'e}chet, Maurice},
  journal={Ann. de la Soc. Polonaise de Math.},
  year={1927}
}

@article{de1977limit,
  title={Limit theory for multivariate sample extremes},
  author={De Haan, Laurens and Resnick, Sidney I},
  journal={Zeitschrift f{\"u}r Wahrscheinlichkeitstheorie und verwandte Gebiete},
  volume={40},
  number={4},
  pages={317--337},
  year={1977},
  publisher={Springer}
}

@book{de2006extreme,
  title={Extreme value theory: an introduction},
  author={De Haan, Laurens and Ferreira, Ana},
  year={2006},
  publisher={Springer}
}

@article{gnedenko1943distribution,
  title={Sur la distribution limite du terme maximum d'une serie aleatoire},
  author={Gnedenko, Boris},
  journal={Annals of mathematics},
  volume={44},
  number={3},
  pages={423--453},
  year={1943},
  publisher={JSTOR}
}

@article{balkema1977max,
  title={Max-infinite divisibility},
  author={Balkema, August A and Resnick, Sidney I},
  journal={Journal of Applied Probability},
  volume={14},
  number={2},
  pages={309--319},
  year={1977},
  publisher={Cambridge University Press}
}

@inproceedings{fisher1928limiting,
  title={Limiting forms of the frequency distribution of the largest or smallest member of a sample},
  author={Fisher, Ronald Aylmer and Tippett, Leonard Henry Caleb},
  booktitle={Mathematical proceedings of the Cambridge philosophical society},
  volume={24},
  number={2},
  pages={180--190},
  year={1928},
  organization={Cambridge University Press}
}

@article{von1936distribution,
  title={La distribution de la plus grande de n valuers},
  author={Von Mises, Richard},
  journal={Rev. math. Union interbalcanique},
  volume={1},
  pages={141--160},
  year={1936}
}

@article{balkema1974residual,
	Author = {Balkema, August A and de Haan, Laurens},
	Journal = {The Annals of probability},
	Pages = {792--804},
	Publisher = {JSTOR},
	Title = {Residual life time at great age},
	Year = {1974}, 
	doi = {https://doi.org/10.1214/aop/1176996548},}

@article{katz2002statistics,
	Author = {Katz, Richard W and Parlange, Marc B and Naveau, Philippe},
	Journal = {Advances in water resources},
	Number = {8-12},
	Pages = {1287--1304},
	Publisher = {Elsevier},
	Title = {Statistics of extremes in hydrology},
	Volume = {25},
	Year = {2002}, 
	doi = {https://doi.org/10.1016/S0309-1708(02)00056-8}}

@book{embrechts2013modelling,
	Author = {Embrechts, Paul and Klüppelberg, Claudia and Mikosch, Thomas},
	Publisher = {Springer Science \& Business Media},
	Title = {Modelling extremal events: for insurance and finance},
	Volume = {33},
	Year = {2013}}

@article{davison2023tail,
  title={Tail risk inference via expectiles in heavy-tailed time series},
  author={Davison, Anthony C and Padoan, Simone A and Stupfler, Gilles},
  journal={Journal of Business \& Economic Statistics},
  volume={41},
  number={3},
  pages={876--889},
  year={2023},
  publisher={Taylor \& Francis}
}

@article{legrand2023joint,
  author = {Juliette Legrand and Pierre Ailliot and Philippe Naveau and Nicolas Raillard},
title = {{Joint stochastic simulation of extreme coastal and offshore significant wave heights}},
volume = {17},
journal = {The Annals of Applied Statistics},
number = {4},
publisher = {Institute of Mathematical Statistics},
pages = {3363 -- 3383},
year = {2023},
doi = {10.1214/23-AOAS1766}
}

@article{coles1999,
  title={Dependence measures for extreme value analyses},
  author={Coles, Stuart and Heffernan, Janet and Tawn, Jonathan},
  journal={Extremes},
  volume={2},
  pages={339--365},
  year={1999},
  publisher={Springer}
}

@book{mcneil2015quantitative,
  title={Quantitative risk management: concepts, techniques and tools-revised edition},
  author={McNeil, Alexander J and Frey, R{\"u}diger and Embrechts, Paul},
  year={2015},
  publisher={Princeton university press}
}

@article{yamai2005value,
  title={Value-at-risk versus expected shortfall: A practical perspective},
  author={Yamai, Yasuhiro and Yoshiba, Toshinao},
  journal={Journal of Banking \& Finance},
  volume={29},
  number={4},
  pages={997--1015},
  year={2005},
  publisher={Elsevier}
}

@book{coles2001,
  title={An introduction to statistical modeling of extreme values},
  author={Coles, Stuart},
  volume={208},
  year={2001},
  publisher={Springer}
}

@article{Josaphat2021,
title = {Dependent conditional value-at-risk for aggregate risk models},
journal = {Heliyon},
volume = {7},
number = {7},
pages = {e07492},
year = {2021},
issn = {2405-8440},
doi = {https://doi.org/10.1016/j.heliyon.2021.e07492},
url = {https://www.sciencedirect.com/science/article/pii/S2405844021015954},
author = {Bony Parulian Josaphat and Khreshna Syuhada},
}

@inproceedings{sklar1959fonctions,
  title={Fonctions de r{\'e}partition {\`a} n dimensions et leurs marges},
  author={Sklar, M},
  booktitle={Annales de l'ISUP},
  volume={8},
  number={3},
  pages={229--231},
  year={1959}
}

@article{ferreira2014generalized,
  title={The generalized Pareto process; with a view towards application and simulation},
  author={Ferreira, Ana and De Haan, Laurens},
  journal={Bernoulli},
  volume={20},
  number={4},
  pages={1717--1737},
  year={2014}
}

@article{de2024bootstrapping,
  title={Bootstrapping extreme value estimators},
  author={de Haan, Laurens and Zhou, Chen},
  journal={Journal of the American Statistical Association},
  volume={119},
  number={545},
  pages={382--393},
  year={2024},
  publisher={Taylor \& Francis}
}

@article{rockafellar2000optimization,
  title={Optimization of conditional value-at-risk},
  author={Rockafellar, R Tyrrell and Uryasev, Stanislav and others},
  journal={Journal of risk},
  pages={21--42},
  year={2000},
  publisher={Citeseer}
}

@article{kiriliouk2019peaks,
  title={Peaks over thresholds modeling with multivariate generalized Pareto distributions},
  author={Kiriliouk, Anna and Rootz{\'e}n, Holger and Segers, Johan and Wadsworth, Jennifer L},
  journal={Technometrics},
  volume={61},
  number={1},
  pages={123--135},
  year={2019},
  publisher={Taylor \& Francis}
}

@article{pickands1975statistical,
  title={Statistical inference using extreme order statistics},
  author={Pickands III, James},
  journal={the Annals of Statistics},
  pages={119--131},
  year={1975},
  volume={3},
  number={1},
  publisher={JSTOR}
}

@article{cai2015estimation,
  title={Estimation of the marginal expected shortfall: the mean when a related variable is extreme},
  author={Cai, Juan-Juan and Einmahl, John HJ and Haan, Laurens and Zhou, Chen},
  journal={Journal of the Royal Statistical Society Series B: Statistical Methodology},
  volume={77},
  number={2},
  pages={417--442},
  year={2015},
  publisher={Oxford University Press}
}

@article{goegebeur2024dependent,
  title={Dependent conditional tail expectation for extreme levels},
  author={Goegebeur, Yuri and Guillou, Armelle and Qin, Jing},
  journal={Stochastic Processes and their Applications},
  pages={104330},
  year={2024},
  publisher={Elsevier}
}

@article{rootzen2018multivariate2,
  title={Multivariate generalized Pareto distributions: Parametrizations, representations, and properties},
  author={Rootz{\'e}n, Holger and Segers, Johan and Wadsworth, Jennifer L},
  journal={Journal of Multivariate Analysis},
  volume={165},
  pages={117--131},
  year={2018},
  publisher={Elsevier}
}

@article{artzner1999coherent,
  title={Coherent measures of risk},
  author={Artzner, Philippe and Delbaen, Freddy and Eber, Jean-Marc and Heath, David},
  journal={Mathematical finance},
  volume={9},
  number={3},
  pages={203--228},
  year={1999},
  publisher={Wiley Online Library}
}

@article{mcneil2000estimation,
  title={Estimation of tail-related risk measures for heteroscedastic financial time series: an extreme value approach},
  author={McNeil, Alexander J and Frey, R{\"u}diger},
  journal={Journal of empirical finance},
  volume={7},
  number={3-4},
  pages={271--300},
  year={2000},
  publisher={Elsevier}
}

@manual{basel2013,
title={{F}undamental review of The Trading Book: A revised market risk framework},
author={{Basel Committee on Banking Supervision}},
organization={{Basel Committee on Banking Supervision}},
year={2013},
shortauthor = {BCBS},  
sortname = {BCBS},  
publisher={Bank for International Settlements},
url={https://www.bis.org/publ/bcbs265.pdf}
}

@article{bickel1981,
  title={Some asymptotic theory for the bootstrap},
  author={Bickel, Peter J and Freedman, David A},
  journal={The annals of statistics},
  volume={9},
  number={6},
  pages={1196--1217},
  year={1981},
  publisher={Institute of Mathematical Statistics}
}

@article{grundemann2023extreme,
  title={Extreme precipitation return levels for multiple durations on a global scale},
  author={Gr{\"u}ndemann, Gaby J and Zorzetto, Enrico and Beck, Hylke E and Schleiss, Marc and Van de Giesen, Nick and Marani, Marco and van der Ent, Ruud J},
  journal={Journal of Hydrology},
  volume={621},
  pages={129558},
  year={2023},
  publisher={Elsevier}
}

@book{nelsen2006introduction,
  title={An introduction to copulas},
  author={Nelsen, Roger B},
  year={2006},
  publisher={Springer}
}

@article{rootzen2006multivariate,
  title={Multivariate generalized Pareto distributions},
  author={Rootz{\'e}n, Holger and Tajvidi, Nader},
  journal={Bernoulli},
  volume={12},
  number={5},
  pages={917--930},
  year={2006},
  publisher={Bernoulli Society for Mathematical Statistics and Probability}
}

@article{nadarajah2014estimation,
  title={Estimation methods for expected shortfall},
  author={Nadarajah, Saralees and Zhang, Bo and Chan, Stephen},
  journal={Quantitative Finance},
  volume={14},
  number={2},
  pages={271--291},
  year={2014},
  publisher={Taylor \& Francis}
}

@article{singh2013extreme,
  title={Extreme market risk and extreme value theory},
  author={Singh, Abhay K and Allen, David E and Robert, Powell J},
  journal={Mathematics and computers in simulation},
  volume={94},
  pages={310--328},
  year={2013},
  publisher={Elsevier}
}

\end{document}